\documentclass[preprint,10pt]{elsarticle_mod}
\usepackage{url}
\usepackage{amsmath}
\usepackage{amssymb}
\usepackage{amsthm}
\usepackage{graphics}
\usepackage{graphicx}
\usepackage{times}
\usepackage{algorithm,algorithmic}

\newcommand{\E}[1]{\mathbf{E}\left[#1\right]}

\newcommand{\var}[1]{\mathbf{Var}\left[#1\right]}
\newcommand\bigforall{\mbox{\Large $\mathsurround0pt\forall$}}

\newtheorem{theorem}{Theorem}

\newtheorem{Assumption}[theorem]{Assumption}

\newtheorem{Lemma}[theorem]{Lemma}

\urldef{\mailsa}\path|rafal.kapelko@pwr.edu.pl|

\begin{document}
\begin{frontmatter}
\title{On the expected moments between two identical random processes with application to sensor network
}
\author[pwr]{Rafa\l{} Kapelko\corref{cor1}}
\ead{rafal.kapelko@pwr.edu.pl}
\fntext[pwrfootnote]{Research supported by grant nr 0401/0086/16}
\cortext[cor1]{Corresponding author at: Department of Computer Science,
Faculty of Fundamental Problems of Technology, Wroc{\l}aw University of Science and Technology, 
 Wybrze\.{z}e Wyspia\'{n}skiego 27, 50-370 Wroc\l{}aw, Poland. Tel.: +48 71 320 33 62; fax: +48 71 320 07 51.}
\address[pwr]{ Department of Computer Science,
Faculty of Fundamental Problems of Technology, Wroc{\l}aw University of Science and Technology, Poland}
\begin{abstract}
We give a closed analytical formula for expected absolute difference to the power $a$ between two identical general random processes, when $a$ is an even positive number.
The following identity is valid
$$\E{|X_k-Y_k|^a}=\frac{a!\left(\var{\xi_1}\right)^{\frac{a}{2}}}{\left(\frac{a}{2}\right)!}\frac{k^{\frac{a}{2}}}{\lambda^a}+\frac{O\left(k^{\frac{a}{2}-1}\right)}{\lambda^a}$$
(see Theorem \ref{thm:mainclosedbe}).

As an application to sensor network we prove that the  optimal transportation cost to the power $b>0$ of the maximal random bicolored matching
with edges $\{X_k,Y_k\}$ is in
$\frac{\Theta\left(n^{\frac{b}{2}+1}\right)}{{\lambda}^b}$
when $b \ge 2,$ and in
$\frac{O\left(n^{\frac{b}{2}+1}\right)}{{\lambda}^b}$ when $0< b < 2.$ 
\end{abstract}
\end{frontmatter}

The aim of the note is to study the problem of the expected absolute difference to a power $b>0$ between two identical general random processes. 
We define general random process as follows.
\begin{Assumption}[general random process]
\label{assumption:first}
Fix $b>0.$ Let $c$ be the smallest even integer  greater than or equal to b.
Consider two identical and independent sequences $\left\{\xi_i\right\}_{i\ge 1},$ $\left\{\tau_i\right\}_{i\ge 1}$  of  identically distributed positive, absolutely continuous random variables.
Assume that
\begin{equation}
\label{eq:motiv01} 
\E{\xi_i}=\E{\tau_i}=1,
\end{equation}
\begin{equation}
\label{eq:motiv02}
\E{\xi^p_i}=\E{\tau^p_i}\le C_c,\,\,\,\,\,\,p\in\{2,3,\dots, c\} 
\end{equation}
for some constans $C_c$ independent on $\lambda>0,$
\begin{align}
\label{eq:motiv03}
\nonumber &\bigforall_{
p_i\in\mathbf{N},\,\,\,
2\le p_1+p_2+\dots +p_j\le c
}\\
\nonumber &\,\,\,\,\,\,\,\,\,\,\,\, \Big(\E{\xi^{p_1}_{i_1}\xi^{p_2}_{i_2}\dots\xi^{p_j}_{i_j}}=\E{\xi^{p_1}_{i_1}}\E{\xi^{p_2}_{i_2}}\dots\E{\xi^{p_j}_{i_j}},\\
&\,\,\,\,\,\,\,\,\,\,\,\, \E{\tau^{p_1}_{i_1}\tau^{p_2}_{i_2}\dots\tau^{p_j}_{i_j}}=\E{\tau^{p_1}_{i_1}}\E{\tau^{p_2}_{i_2}}\dots\E{\tau^{p_j}_{i_j}}\Big).
\end{align}
Let $X_k=\frac{1}{\lambda}\sum_{i=1}^k\xi_i,$ $Y_k=\frac{1}{\lambda}\sum_{i=1}^k\tau_i.$
\end{Assumption}
We are interested in the moments (of each $b>0$)
$$
\E{|X_{k+r}-Y_k|^b},\,\,\,\text{when}\,\,\,k\ge 1, r\ge 0, \text{(see Assumption \ref{assumption:first})}.
$$
More importantly, our work is related to the paper \cite{dam_2014} where the author studied the 
expected absolute difference of the arrival times between two identitcal and independent Poisson processes
with respective arrival times $P_1,P_2,\dots$ and $Q_1,Q_2,\dots$ on a line and derived a closed form formula for the 
\begin{equation}
\label{eq:kranakis}
\E{|P_{k+r}-Q_k|},\,\,\,\text{for any}\,\,\,k\ge 1, r\ge 0.
\end{equation}
The paper \cite{dam_2014} treats only the very special case when $P_k,$ $Q_k$ obeys the gamma distribution with parameters $k, \lambda.$

The following open problem was proposed in \cite{dam_2014} to study the moments
$$\E{|P_{k+r}-Q_k|^b},$$ 
where $b>0$ is fixed for more general random processes. 

We extend the work in \cite{dam_2014} by considering the expected absolute difference
to all exponents $b>0$ between two identical and independent general random processes
and thus solve the open problem. The main advantage of our approach is to derive closed form asymptotic formulas for the moments without use of any specific density function (gamma distribution) 
for a wide class of distributions.

As another point of motivation for studying these expected absolute difference to exponent arise in sensor networks. 
We consider two sequences $\left\{X_i\right\}_{i=1}^{n},$ $\left\{Y_i\right\}_{i=1}^{n}$ (see Asumption (\ref{assumption:first})). The sensors in 
$X_1,X_2,\dots ,X_n$ are colored black and the sensors in $Y_1,Y_2,\dots,Y_n$ are colored white.
We are interested in expected minimum sum of length to exponent of a maximal bicolored matching 
(the vertices of each matching edge have different colors).

The cost of sensor movement has been studied extensively in the research community (e.g., see \cite{ajtai_84, kumar2005, sajal2008, cortes2012, talagrand_2014, siamcontrol2015, kapelkokranakisIPL, KK_2016_cube}).
The book \cite{talagrand_2014}  addresses the matching theorems for $N$ random variables independently uniformly distributed in the $d-$dimensional unit cube
$[0,1]^d,$ where $d\ge 2.$ The authors of \cite{kapelkokranakisIPL} deal with covering of
the unit interval with uniformly and independently at random placed sensors
when the cost of movement of sensors is proportional to some (fixed) power $a>0.$
This paper \cite{siamcontrol2015} regards the problem of optimally placing unreliable sensors in a one-dimensional environment.

Closely related to our work is the paper \cite{kapelkogamma} where the author considers two  identical and independent Poisson processes with arrival rate $\lambda>0$ 
and respective arrival times $X_1,X_2,\dots$ and $Y_1,Y_2,\dots$ on a line
and gives a closed analytical formula for the $\E{|X_{k+r}-Y_k|^a}, $ for any integer $k\ge 1, r\ge 0$ and $a\ge 1.$
\section{Main results}
\label{tight:sec}
Fix $b>0.$ 
In this section the expected absolute difference to the power $b$ between two identical and independent  general random processes is analyzed.

Firstly, we derive closed form formula for expected absolute difference to the power $a$ between two identical general random processes, when $a$ is an even positive integer.
We proof Theorem \ref{thm:mainclosedbe}.
\begin{theorem} 
\label{thm:mainclosedbe} 
Let us fix an even positive integer $a.$
Let Assumption \ref{assumption:first} hold for $b:=a$ and let $k>\frac{a}{2}.$ 
Then the following identity is valid
$$
\E{|X_k-Y_k|^a}=\frac{a!\left(\var{\xi_1}\right)^{\frac{a}{2}}}{\left(\frac{a}{2}\right)!}\frac{k^{\frac{a}{2}}}{\lambda^a}+\frac{O\left(k^{\frac{a}{2}-1}\right)}{\lambda^a}.$$
\end{theorem}
The general strategy of our combinatorial proof of Theorem \ref{thm:mainclosedbe} is the following.
Applying multinomial theorem we write $\E{|X_k-Y_k|^a}$ as the sum 
(see Equation (\ref{align:first})). Next, we make an important observation that $\mathbf{E}\left[(\xi_i-\tau_i)^{2d+1}\right]=0$
(see Equation (\ref{align:second})). Using this, we rewrite $\E{|X_k-Y_k|^a}$ as the sums (\ref{align:3}) and (\ref{align:5}).
Finally, the asymptotic depends on the expression given by the first sum (see Equation (\ref{align:3})), while the second sum (see Equation (\ref{align:5})) is negligible.
\begin{proof}
Fix an even positive integer $a.$ Assume that $k>\frac{a}{2}.$ 

Firstly, combining together multinomial theorem, Equation (\ref{eq:motiv03}) as well as Assumption \ref{assumption:first} we deduce that 
\begin{align}
\label{align:first}
&\nonumber\E{|X_k-Y_k|^a}=\E{(X_k-Y_k)^a}\\
&\,\,\,\,\,\,=\sum_{S}
\nonumber\frac{a!}{(l_1)!(l_2)!\dots(l_k)!}\frac{1}{\lambda^a}\E{\prod_{i=1}^k\left(\xi_i-\tau_i\right)^{l_i}}\\
&\,\,\,\,\,\,=\sum_{S}
\frac{a!}{(l_1)!(l_2)!\dots(l_k)!}\frac{1}{\lambda^a}\prod_{i=1}^k\E{\left(\xi_i-\tau_i\right)^{l_i}},
\end{align}
where 
$$S=\{(l_1,l_2,\dots l_k)\in\mathbb{N}^k: l_1+l_2+\dots+l_k=a\}.$$
Let $d$ be natural number. Using Assumption \ref{assumption:first} and the basic binomial identity\\ 
$\binom{2d+1}{j}(-1)^{2d+1-j}=-\binom{2d+1}{2d+1-j}(-1)^j$ we have
\begin{align}
\label{align:second}
\nonumber\mathbf{E}&\left[(\xi_i-\tau_i)^{2d+1}\right]\\
\nonumber &=\sum_{j=0}^{2d+1}\binom{2d+1}{j}\E{\xi^j_i}(-1)^{2d+1-j}\E{\tau^{2d+1-j}_i}\\
\nonumber &=\sum_{j=0}^{2d+1}\binom{2d+1}{j}\E{\tau^j_i}(-1)^{2d+1-j}\E{\tau^{2d+1-j}_i}\\
\nonumber &=\sum_{j=0}^{d}\E{\tau^j_i}\E{\tau^{2d+1-j}_i} \binom{2d+1}{j}(-1)^{2d+1-j}\\
\nonumber &+\sum_{j=0}^{d}\E{\tau^j_i}\E{\tau^{2d+1-j}_i}\binom{2d+1}{2d+1-j}(-1)^{j}\\
&=0
\end{align}
Combining together (\ref{align:first}) and (\ref{align:second}) we deduce that
\begin{equation}
\label{equation:first}
\E{|X_k-Y_k|^a}=
\sum_{S_1
}
\frac{a!}{(l_1)!(l_2)!\dots(l_k)!}\frac{1}{\lambda^a}\prod_{i=1}^k\E{\left(\xi_i-\tau_i\right)^{l_i}},
\end{equation}
where
\begin{align*}
S_1=&\{(l_1,l_2,\dots l_k)\in\mathbb{N}^k: l_1+l_2+\dots+l_k=a,\,\,\,\\
&\,\,\,\,\,l_i\,\,\,\text{are even for}\,\,\, i=1,2,\dots,k\}.
\end{align*}
Observe that
\begin{equation}
\label{eq:suma}
S_1=S_2\cup S_3,
\end{equation}
\begin{align*}
 \nonumber S_2=&\{(l_1,l_2,\dots l_k)\in\mathbb{N}^k: l_1+l_2+\dots+l_k=a,\,\,\,\\
 &\,\,\,\,\,l_i\in\{0,2\}\,\,\,\text{for}\,\,\,i=1,2,\dots,k\},
\end{align*}
\begin{align*}
\nonumber S_3=&\{(l_1,l_2,\dots l_k)\in\mathbb{N}^k:  l_1+l_2+\dots+l_k=a,\,\,\\
&\,\,\,\,\,l_i\,\text{are even for}\,\,\,i=1,2,\dots,k,
\,\,\,\,\,\exists i\,(l_i\neq 2)\},
\end{align*}
\begin{equation}
 \label{eq:cardinality}
 |S_2|=\binom{k}{\frac{a}{2}},\,\,\,\,\,\,|S_3|=O\left(k^{\frac{a}{2}-1}\right).
\end{equation}
Let $f(t)$ be the probability density function of the random variables $\xi_i,\,\,\tau_i.$ We use H\"older's inequalities
with $p\in\{2,3,\dots,a\},\,\,\,$ $q=\frac{p-1}{p}$ and  get the following sharp inequalities
\begin{equation}
\label{eq:holder}
\int_0^{\infty}tf(t)dt<\left(\int_0^{\infty}t^pf(t)dt\right)^{1/p}\left(\int_0^{\infty}f(t)dt\right)^{1/q}.
\end{equation}
Putting together Inequality (\ref{eq:holder}) and Equality (\ref{eq:motiv01}) in Assumption \ref{assumption:first} we deduce that
\begin{equation}
\label{eq:sharp}
\E{\xi^p_i}=\E{\tau^p_i}>1,\,\,\,\text{when}\,\,\,p\in\{2,3.\dots,a\}.
\end{equation}
Observe that 
\begin{equation}
\label{eq:dus}
\E{\left(\xi_i-\tau_i\right)^{2}}=2\left(\E{\xi_1^2}-1\right)=2\left(\var{\xi_1}\right)>0.
\end{equation}
Together, (\ref{eq:cardinality}), (\ref{eq:dus}) and (\ref{eq:sharp}) imply
\begin{align}
\label{align:3}
\nonumber\sum_{S_2}&\frac{a!}{(l_1)!(l_2)!\dots(l_k)!}\frac{1}{\lambda^a}\prod_{i=1}^k\E{\left(\xi_i-\tau_i\right)^{l_i}}\\
\nonumber&\,\,\,=\frac{a!\left(\var{\xi_1}\right)^{\frac{a}{2}}}{{\lambda^a}}|S_2|\\
&\,\,\,=\frac{a!\left(\var{\xi_1}\right)^{\frac{a}{2}}}{\left(\frac{a}{2}\right)!}\frac{k^{\frac{a}{2}}}{\lambda^a}+\frac{O\left(k^{\frac{a}{2}-1}\right)}{\lambda^a}.
\end{align}
Using Inequality  (\ref{eq:motiv02}) in Assumption \ref{assumption:first} we have 
\begin{align}
\label{align:4}
\nonumber\E{\left(\xi_i-\tau_i\right)^{l_i}}&=\E{\left|\xi_i-\tau_i\right|^{l_i}}\le \E{\left(\left|\xi_i\right|+\left|\tau_i\right|\right)^{l_i}}\\
\nonumber&=\sum_{j=0}^{l_i}\binom{l_i}{j}\E{\xi_i^j}\E{\tau_i^{l_i-j}}\\
&\le C_a^2\sum_{j=0}^{l_i}\binom{l_i}{j}=C_a^2 2^{l_i}.
\end{align}
Together, (\ref{align:4}), (\ref{eq:cardinality}) and $C_a>1$ (see (\ref{eq:sharp})) imply
\begin{align}
\label{align:5}
\nonumber\sum_{S_3}&
\frac{a!}{(l_1)!(l_2)!\dots(l_k)!}\frac{1}{\lambda^a}\prod_{i=1}^k\E{\left(\xi_i-\tau_i\right)^{l_i}}\\
&\le a!2^a C_a^{2a}\frac{|S_3|}{\lambda^a}
=\frac{O\left(k^{\frac{a}{2}-1}\right)}{\lambda^a}.
\end{align}
Finally, combining together (\ref{equation:first}), (\ref{eq:suma}), (\ref{align:3}), (\ref{align:5})
finishes the proof of  Theorem \ref{thm:mainclosedbe}. 
\end{proof}
The next results supports our earlier result whereby the expected absolute difference to the power $a$  between two identical and independent random processes,
remains in $\frac{\Theta\left(k^{\frac{a}{2}}\right)}{\lambda^a}$ provided that $r=o\left(k^{\frac{1}{2}}\right).$
\begin{theorem} 
\label{thm:maincl} 
Let us fix an even positive integer $a.$  
Let Assumption \ref{assumption:first} hold for $b:=a.$
If $r=o\left(k^{\frac{1}{2}}\right)$
then
$$
\E{|X_{k+r}-Y_k|^a}=\frac{\Theta\left(k^{\frac{a}{2}}\right)}{\lambda^a}.
$$
\end{theorem}
\begin{proof}
Firstly, we apply multinomial theorem, Equation (\ref{eq:motiv03}) and get
\begin{align*}
& \E{|X_{k+r}-X_k|^a}=\E{\left|\sum_{i=1}^r X_{k+i}\right|^a}=\E{\left(\sum_{i=1}^r X_{k+i}\right)^a}\\
 &\,\,\,\,\,\,= \sum_{l_1+l_2+\dots l_r=a}
\frac{a!}{(l_1)!(l_2)!\dots(l_r)!}\frac{1}{\lambda^a}\prod_{i=1}^r\E{\xi_{k+i}^{l_i}}.
\end{align*}
Using Inequality (\ref{eq:motiv02}) in Assumption \ref{assumption:first} and (\ref{eq:sharp}) we have
$$
\prod_{i=1}^r\E{\xi_{k+i}^{l_i}}\le C_a^a.
$$
Hence
\begin{align*}
&\E{|X_{k+r}-X_k|^a}\\
&\,\,\,\,\,\le\frac{C_a^a}{\lambda^a}\sum_{l_1+l_2+\dots l_r=a}
\frac{a!}{(l_1)!(l_2)!\dots(l_r)!} =\frac{C_a^a}{\lambda^a}r^a.
\end{align*}
Since $r=o\left(k^{\frac{1}{2}}\right),$ we have
\begin{equation}
\label{eq:alfa}
\E{|X_{k+r}-X_k|^a}=\frac{o\left(k^{\frac{a}{2}}\right)}{\lambda^a}.
\end{equation}
Combining together assumption $r=o\left(k^{\frac{1}{2}}\right)$ and the result of Theorem  \ref{thm:mainclosedbe} for $k:=k+r$ we easily deduce that
\begin{equation}
\label{eq:beta}
\E{|X_{k+r}-Y_{k+r}|^a}=\frac{\Theta\left(k^{\frac{a}{2}}\right)}{\lambda^a}.
\end{equation}
Notice that
\begin{align}
\label{eq:gamma}
\nonumber&|x+y|^a\le\left(|x|+|y|\right)^a\\
&\,\,\,\,\,\le 2^{a-1}\left(|x|^a+|y|^a\right)\,\,\,\,\,\text{when}\,\,\,\,\,a\ge 1,\,\, x,y\in\mathbb{R}.
\end{align}
This inequality follows from the fact that $f(x)=x^a$  is convex over $\mathbb{R_+}$ for $a\ge 1.$
Applying (\ref{eq:gamma}) for $x:=X_{k}-Y_k,$ $y:=X_{k+r}-X_{k}$ we get
\begin{align}
\label{eq:saam}
\nonumber&\E{|X_{k+r}-Y_k|^a}\\
&\,\,\,\,\,\le 2^{a-1}\left(\E{|X_k-Y_k|^a}+\E{|X_{k+r}-X_k|^a}\right).
\end{align}
Combining together (\ref{eq:saam}), the result of Theorem \ref{thm:mainclosedbe} and Equation (\ref{eq:alfa}) we have the desired upper bound
$$
\E{|X_{k+r}-Y_k|^a}=\frac{O\left(k^{\frac{a}{2}}\right)}{\lambda^a}\,\,\,\,\,\,\text{if}\,\,\,\,\,\,r=o\left(k^{\frac{1}{2}}\right).
$$
Next, applying (\ref{eq:gamma}) for $x:=X_{k+r}-Y_k,$ $y:=X_{k}-X_{k+r}$ we have
\begin{align}
\label{eq:sbam}
\nonumber&\E{|X_{k}-Y_k|^a}\\
&\,\,\,\,\,\le 2^{a-1}\left(\E{|X_{k+r}-Y_k|^a}+\E{|X_{k}-X_{k+r}|^a}\right).
\end{align}
Together (\ref{eq:sbam}), the result of Theorem \ref{thm:mainclosedbe} and Equation (\ref{eq:alfa}) imply the lower bound
$$
\E{|X_{k+r}-Y_k|^a}=\frac{\Omega\left(k^{\frac{a}{2}}\right)}{\lambda^a}\,\,\,\,\,\,\text{if}\,\,\,\,\,\,r=o\left(k^{\frac{1}{2}}\right).
$$
This is sufficient to complete the proof of Theorem \ref{thm:maincl}.
\end{proof}
The next theorem extends our Theorem \ref{thm:maincl} to real-valued exponents.
In the proof of Theorem \ref{thm:maincles} we combine together Jensen's inequality and the results of Theorem \ref{thm:maincl}

Let us recall Jensen's inequality for expectations. If $f$ is a convex function, then
\begin{equation}
\label{eq:jensen}
f\left( \mathbf{E}[X]\right)\le  \mathbf{E}\left[f(X)\right]
\end{equation}
provided the expectations exists (see \cite[Proposition 3.1.2]{ross2002}).
\begin{theorem} 
\label{thm:maincles} 
Fix $b>0.$
Let Assumption \ref{assumption:first} hold. 
If $r=o\left(k^{\frac{1}{2}}\right),$
then
$$
 \E{|X_{k+r}-Y_k|^b}=\begin{cases} \frac{\Theta\left(k^{\frac{b}{2}}\right)}{{\lambda}^b}\,\,\, &\mbox{if}\,\,\, b \ge 2 \\
\frac{O\left(k^{\frac{b}{2}}\right)}{{\lambda}^b}\,\,\, & \mbox{if }\,\,\, 0< b < 2. \end{cases}  
$$ 
\end{theorem}
\begin{proof}
Assume that $b>0.$ Let $c$ be the smallest even integer  greater than or equal to b. 

First we prove the upper bound. 
We use Jensen's inequality (see (\ref{eq:jensen})) for $X=|X_{k+r}-Y_k|^b$ and
$f(x)=x^{\frac{c}{b}}$ and get
\begin{equation}
\label{eq:jensen01}
\left(\E{|X_{k+r}-Y_k|^b}\right)^{\frac{c}{b}}\le\E{|X_{k+r}-Y_k|^{c}}.
\end{equation}
Putting together Theorem \ref{thm:maincl} for $a:=c$ and Inequality (\ref{eq:jensen01}) we deduce that
$$
\E{|X_{k+r}-Y_k|^b}\le\left(\frac{\Theta\left(k^{\frac{c}{2}}\right)}{\lambda^{c}}\right)^{\frac{b}{c}}
=\frac{\Theta\left(k^{\frac{b}{2}}\right)}{\lambda^b}.
$$
This proves the upper bound.

Next we prove the lower bound. 
Assume that $b\ge 2.$ We apply Jensen's inequality (see (\ref{eq:jensen})) for $X=|X_{k+r}-Y_k|^2$ and
$f(x)=x^{\frac{b}{2}}$ and have
\begin{equation}
\label{eq:jensen02}
\left(\E{|X_{k+r}-Y_k|^2}\right)^{\frac{b}{2}}\le\E{|X_{k+r}-Y_k|^{b}}.
\end{equation}
Combining together Theorem \ref{thm:maincl} for $a:=2$ and Inequality (\ref{eq:jensen02}) we get
$$
\E{|X_{k+r}-Y_k|^b}\ge\left(\frac{\Theta\left(k\right)}{\lambda^2}\right)^{\frac{b}{2}}
=\frac{\Theta\left(k^{\frac{b}{2}}\right)}{\lambda^b}.
$$
This is enough to prove the lower bound and completes the proof of Theorem \ref{thm:maincles}.
\end{proof}
\section{Application to sensor networks}
In this section, we consider the optimal transportation cost to the power $b$ of the maximal random bicolored matching, when
$b>0.$ 

Assume that sensors are initially placed according to two general random processes.
Let Assumption \ref{assumption:first} hold. 
The sensors in 
$X_1,X_2,\dots ,X_n$ are colored black and the sensors in $Y_1,Y_2,\dots,Y_n$ are colored white.

We would like to find the maximal bicolored matching $M$ so as to:
\begin{enumerate}
\item[(1)] 
for every pair of sensors $\{X_k,Y_l\}\in M,$  the sensors $X_k,$ $Y_k$ have different colors,
\item[(2)] the expected transportation cost to the power $b>0$  defined as 
$$T_b(M):=\sum_{\{X_k,Y_l\}\in M} \E{|X_k-Y_l|^b}$$
is minimized.
\end{enumerate}
Firstly, we observe that the minimal transportation cost to the power $b$ is attained by the maximal matching with edges $\{X_k,Y_k\}$
for $k=1,2,\dots,n.$ 
\begin{Lemma}
\label{thm:last}
 Fix $b\ge 0.$ Let $M_{\text{opt}}$ be the maximal matching with edges $\{X_k,Y_k\},$ for $k=1,2,\dots,n.$ Then for all matchings $M$ we have
$$T_b(M)\ge T_b\left(M_{\text{opt}}\right).$$
\end{Lemma}
\begin{proof}
The proof is essentially the same as the proof of  \cite[Lemma 5]{dam_2014}. 
\end{proof}
We are now ready to analyze the maximal matching with edges $\{X_k,Y_k\}$ for $k=1,2,\dots,n.$  Applying Theorem \ref{thm:maincles} from the previous section
we can prove the following theorem.
\begin{theorem}
\label{thm:app1}
 Fix $b>0.$ If $M_{\text{opt}}$ denotes the maximal matching with edges $\{X_k,Y_k\}$ for $k=1,2,\dots,n,$ then
$$
T_b(M_{\text{opt}})=\begin{cases} \frac{\Theta\left(n^{\frac{b}{2}+1}\right)}{{\lambda}^b}\,\,\, &\mbox{when}\,\,\, b \ge 2 \\
\frac{O\left(n^{\frac{b}{2}+1}\right)}{{\lambda}^b}\,\,\, & \mbox{when }\,\,\, 0< b < 2. \end{cases} 
$$
\end{theorem}
\begin{proof}
First of all, observe that
$$T_b(M_{\text{opt}})=\sum_{k=1}^n\E{|X_k-Y_k|^b}.$$ 
After that, the result of Theorem \ref{thm:app1} follows immediately from  well known identity
$$\sum_{k=1}^nk^{\frac{b}{2}}=\Theta\left(n^{\frac{b}{2}+1}\right)\,\,\,\,\text{when}\,\,\,\,b> 0$$
and Theorem \ref{thm:maincles} for $r:=0.$
\end{proof}

\bibliographystyle{plain}
\bibliography{refs}
\end{document}